\documentclass[opre,nonblindrev]{informs3} 

\OneAndAHalfSpacedXI 

\usepackage{natbib}
\bibpunct[, ]{(}{)}{,}{a}{}{,}%
\def\bibfont{\small}%
\usepackage{booktabs} 

\usepackage[ruled]{algorithm2e} 

\usepackage{ bbm, color, enumerate, amsbsy, amsfonts, amsopn, amssymb,  amsxtra, bezier, graphicx, latexsym, verbatim,
	pictexwd, supertabular, url, dsfont,leftidx, setspace}
\usepackage{mathtools,bm}
\usepackage{multirow}
\usepackage{footnote}
\makesavenoteenv{tabular}
\makesavenoteenv{table}

\newtheorem{eg}{Example}[section]

\newcommand{\mb}[1]{\ensuremath{\boldsymbol{#1}}}

\newcommand{\onee}{\mathbbm{1}}

\DeclarePairedDelimiter{\ceil}{\lceil}{\rceil}

\def\alg{\textsc{PR}}

\TheoremsNumberedThrough     
\ECRepeatTheorems

\EquationsNumberedThrough    

\begin{document}
	
	
	
	\TITLE{
	Periodic Reranking for Online Matching of Reusable Resources 
} 
	
	\ARTICLEAUTHORS{%
		\AUTHOR{	Rajan Udwani}
	\AFF{UC Berkeley, IEOR, 
			\EMAIL{rudwani@berkeley.edu}}
	} 
	
	\ABSTRACT{%
We consider a generalization of the vertex weighted online bipartite matching problem where the offline vertices, called resources, are \emph{reusable}. In particular, when a resource is matched it is unavailable for a deterministic time duration $d$ after which it becomes available for a re-match. Thus, a resource can be matched to many different online vertices over a period of time. 

While recent work on the problem has resolved the asymptotic case where we have large starting inventory (i.e., many copies) of every resource, we consider the (more general) case of \emph{unit inventory} and give the first algorithm that is provably better than the na\"ive greedy approach which has a competitive ratio of (exactly) 0.5. In particular, we achieve a competitive ratio of 0.589 against an LP relaxation of the offline problem. Our algorithm generalizes the classic Ranking and Perturbed Greedy algorithms for online matching, by \emph{reranking} resources over time. While reranking resources frequently has the same worst case performance as greedy, we show that reranking \emph{intermittently} on a \emph{periodic schedule} succeeds in addressing reusability of resources and performs significantly better than greedy in the worst case. 
	}%
	
	
	\maketitle
	
	%
	
	
	\section{Introduction}\label{intro}	
On a platform such as Airbnb, where heterogeneous customers arrive over time, ensuring a good match between property listings (resources) and customers is a challenging task. Of the many challenges, a well studied one stems from the sequential and uncertain (online) nature of demand. There is a wealth of work on policies for pricing and allocating resources when the demand is online (online matching \citep{survey}, online assortments \citep{negin}, pricing and network revenue management \citep{txt}). Most of these works focus on resources that are \emph{not reusable} and may be allocated at most once. However, a feature common to most online platforms in the sharing economy is that resources are \emph{reusable} and a given unit of a resource may be re-allocated several times. This has led to a surge of interest in designing policies for online allocation of reusable resources. Nonetheless, many fundamental questions remain unanswered. In this paper, we study one such question.  

We start by defining the setting of online matching with reusable resources (OMR). Consider a set of reusable resources $I$ with unit inventory and price $(r_i)_{i\in I}$. Requests for resources arrive sequentially and each request is for up to one unit of some subset of resources. Formally, let $T$ denote the set of requests, also called \emph{arrivals}, and let $A$ denote the set of arrival times. Up on arrival of request $t\in T$ at time $a(t)\in A$, we observe set of edges $(i,t)\in E$ incident on $t$, and must make an immediate and irrevocable decision to match $t$ to a neighboring resource or reject the request. At time $a(t)$, we have no knowledge of future arrivals, which could arrive in an adversarial order. When resource $i\in I$ is matched to an arrival, we obtain reward $r_i$ and the resource is unavailable for subsequent arrivals during the next $d$ units of time i.e., if $i$ is matched to $t$, it is unavailable during the interval $(a(t),a(t)+d]$. After this interval, resource $i$ is available for re-match. The duration $d$ is known to us and our objective is to maximize total reward. 
An online algorithm for this problem 
	is evaluated relative to the optimal offline solution to the problem. The offline solution is computed with knowledge of the entire sequence of arrivals $T$ and arrival times $A$. This comparison is quantified through the \emph{competitive ratio}, defined as follows,  
	\[\min_{G}\,\, \frac{\text{Online}(G)}{\text{Offline}(G)},\]
	where $G$ represents an instance of the problem, characterized by sets $I,\, T,$ and $A$ and values $(r_i)_{i\in I}$ and $d$. Online$(G)$ represents the (expected) total reward of a (randomized) online algorithm. We are interested in designing an online algorithm with a strong competitive ratio guarantee for this problem.
	
 OMR is a fundamental generalization of the classic online bipartite matching (OM) problem \citep{kvv}, where resources have identical rewards and each resource can be matched at most once. To see this, observe that when $d\geq a(T)$ and rewards $r_i=r_j\,\forall i,j\in I$, then every resource can be matched to at most one arrival and the objective simplifies to finding the maximum matching. 
	For OM, it is well known that no deterministic algorithm can achieve a competitive ratio better than 0.5 \citep{kvv}. As OMR generalizes OM, this upper bound applies directly. In fact, the greedy algorithm that matches every arrival to an available neighbor (and breaks ties deterministically) is exactly 0.5 competitive for both OM \citep{kvv} and OMR \citep{reuse}. The main algorithm design goal in these problems and their many variations, is to \emph{design an online algorithm that outperforms greedy}.
	
	For the OM problem, \cite{kvv} proposed the Ranking algorithm that randomly ranks resources at the beginning and then matches each arrival to the best ranked neighbor that is available. They showed that Ranking is $(1-1/e)$ competitive for OM and this is the best possible guarantee for any online algorithm. For the OMR problem, 
	no algorithm with guarantee better than 0.5 was known prior to this work. We establish the following result for this problem.
	
	\begin{theorem} \label{ref}
		There is a randomized algorithm for OMR (Algorithm \ref{pr} with $\beta=0.89$), with competitive ratio 0.589.
		\end{theorem}
	\textbf{Upper bound for OMR:} Recall that a $(1-1/e)$ upper bound for OMR follows from the fact that OMR reduces to OM for large $d$. 
	In fact, 
	this upper bound holds for all non-zero values of $d$. 
	To see this, observe that in OMR there can be an arbitrary number of arrivals in any given time duration. 
Therefore,	for any $d>0$, an entire instance(s) of OBM (including the hard ones) can fit into a duration smaller than $d$. More generally, this imples that OMR is equally hard for all finite and non-zero $d$ i.e., if there is an upper bound of $\alpha$ for duration $d'$, one can scale the hard instances to obtain an $\alpha$ upper bound for any other (finite and non-zero) duration $d$. 

	Next, we discuss previous work on this problem. This is followed by a discussion on the significance of the unit inventory setting.

	\subsection{Previous Work}\label{sec:prev}
	Settings with uncertain sequentially arriving demand and reusable resources has received significant interest recently. We start by discussing work that is closest to our setting. A generalization of OMR was first studied by \cite{reuse}. In their setting, 
	\begin{enumerate}[(i)]
		\item  Usage durations are stochastic i.e., when a resource $i\in I$ is allocated, it is used for an independently sampled random duration $d\sim F_i$.
		\item Instead of matching an arrival to a resource, we offer an assortment (i.e., set) of resources to each arrival, and they choose up to one resource from the assortment according to a stochastic choice model that is revealed on arrival.
	\end{enumerate}
	  They showed that the greedy algorithm which offers a revenue maximizing assortment to each arrival, is 0.5 competitive for this general setting. Subsequently,  \cite{feng, feng3} and \cite{full} 
	   considered this setting with the additional structural assumption of \emph{large starting inventory} for every resource i.e., large number of identical copies of each resource. Despite this assumption, $(1-1/e)$ is the best possible guarantee for the problem. \cite{feng, feng3} showed that a classic inventory balancing algorithm, originally proposed for non-reusable resources, 
	  is $(1-1/e)$ competitive for reusable resources with deterministic (but not necessarily identical) usage times. \cite{full} considered the general case of stochastic usage durations (same as \cite{reuse}), and demonstrated that classic approaches fail to improve on the performance of greedy in this more general setting. They 
	  proposed a novel algorithm that accounts for the stochastic nature of reusability by balancing ``effective" inventory in a fluid way, and 
	  achieves the best possible guarantee of $(1-1/e)$ for arbitrary usage distributions.  
	  We note that a parallel stream of work considers the setting of reusable resources with stochastic arrivals \citep{dickerson, RST18, baek, feng2}. For a detailed review of these settings, see \cite{reuse} and \cite{full}. 
	  
	  The setting of unit inventory (considered in this paper) captures the large inventory setting as a special case. When we have multiple copies of a resources, we may treat each copy as a distinct resource with unit inventory. Thus, an algorithm designed for the unit inventory case can be generalized to settings with arbitrary (known) inventory without affecting its performance guarantee (see \cite{reuse} for a formal proof). Prior to this work, the 0.5 guarantee of greedy was the best known result for OMR.

	  For non-reusable resources, there is a wealth of work that improves on greedy and achieves the best possible guarantee of $(1-1/e)$. 
	 Recall, \cite{kvv} introduced the OM problem and showed (among other results) that the Ranking algorithm which randomly ranking resources at the start and then matches every arrival to the best ranked unmatched resource, is $(1-1/e)$ competitive for OM and that this is the best possible guarantee achievable for the problem. 
	 The analysis of Ranking was clarified and considerably simplified by \cite{baum} and \cite{goel2}. In the more general setting where resources have arbitrary rewards $r_i$, \cite{goel} proposed the Perturbed Greedy algorithm and showed that is $(1-1/e)$ competitive for this generalization of OM.  
	 In the large inventory setting, \cite{pruhs} considered the problem of online $b-$matching where the budget of every resource can be more than 1 and showed that as $b\to \infty$, the natural (deterministic) algorithm that balances the budget used across resources is $(1-1/e)$ competitive. 
	 Also in the large inventory setting, \cite{msvv} introduced the Adwords problem which generalizes the OM setting by allowing multi-unit demand. They 
	 gave an online algorithm with guarantee of $(1-1/e)$ for Adwords. 
	 \cite{buchbind} gave a primal-dual analysis for the result of \cite{msvv}. Subsequently, \cite{devanur} proposed the randomized primal-dual framework that can be used to show the aforementioned results in a unified way.
 In addition to these settings, there is a vast body of work on online matching and (non-reusable) resource allocation in stochastic and hybrid/mixed models of arrival. For a comprehensive review of these works, 
 see \cite{survey}. 
 
 		  Finally, in related work, \cite{moharir2015online} introduced a unit inventory model where  
 the arrival sequence is divided into slots and resources are reusable with a (deterministic) usage duration of one slot. At the beginning of each slot, arrivals are sequentially revealed within a short (infinitesimal) amount of time. The number and types of arrivals in a slot is arbitrary. 
A resource can be matched to at most one arrival in each slot and 
 a resource matched in slot $t$ is available for rematch in slot $t+1$.
 Consequently, 
 the decision across slots are independent 
and it can be shown that the classic Ranking algorithm \citep{kvv}, is $(1-1/e)$ competitive. \cite{moharir2015online} proposed several new algorithms, including a $(1-1/e)$ competitive reranking algorithm that samples a new rank for the resources at the beginning of each slot. In fact, they show this result for the more general case where arrivals have heterogeneous match deadlines. Overall, their setting and results are incomparable to ours.
	 	\subsection{Significance of the Unit Inventory Setting}
	 	Reusability of resources is an undeniably important aspect of online platforms such as Airbnb, Upwork, Thumbtack. 
	 	In these settings, each resource is unique and there is often just one unit of inventory per resource. For example, on a platform such as Airbnb, where every listing is a reusable resource, 
	 	a typical listing may be occupied by at most one customer at a time. 
	 	A similar situation arises in case of boutique hotels \citep{sumida}. On platforms such as Upwork and Thumbtack \citep{feng3}, each free-lancing agent can be modeled as a distinct reusable resource that can perform at most one task at a time. 
	 	Settings with small (not necessarily unit) inventory also arise in applications where procuring new inventory is expensive, sales are slow moving, and resources can be reused many times before expiry. For instance, \cite{besbes} consider a setting where resources are rotable spare parts for aircrafts and the starting inventory for most parts is under 10 units (see Figure 9 in \cite{besbes}).  
	 	
	 	In contrast, the large inventory assumption is natural for applications such as cloud computing, where each machine is a resource and the capacity of a machine is the number of jobs it can handle in parallel \citep{full}. 
	 	Another instance where the large inventory assumption is appropriate is make-to-order settings \citep{reuse}, where each production line or machine is a resource and the capacity is measured in the number of units of a good that the machine can manufacture in a given time period.

	
	 
	\section{The Periodic Reranking Algorithm}

\begin{algorithm}[H]
	\SetAlgoNoLine
	\textbf{Inputs:} Set of resources $I$, usage duration $d$, parameter $\beta$\; 
	Let $g(t)=e^{\beta(t-1)}$ and $S=I$\;
	\smallskip
\textbf{Every $d$ time units:} Generate new i.i.d. ranks $y_i\sim U[0,1]\,\, \forall i\in I$\;
\smallskip
	\For{\text{every new arrival } $t$}{
		Update set $S$ by adding resources that returned since arrival $t-1$\;
		Match $t$ to $i^*=\underset{ i\in S,\, (i,t)\in E}{\arg\max}\quad r_i (1-g(y_i))$\;
		$S=S\backslash \{i\}$\;
}	
	\caption{Periodic Reranking (PR)}
	\label{pr}
\end{algorithm}

At the start of the planning horizon, the PR algorithm (independently) samples a random seed $y_i\in U[0,1]$, for every $i\in I$. Using this seed, and a monotonically increasing trade-off function $g:\mathbb{R}\to [0,1]$, the algorithm evaluates \emph{reduced prices} $r_i(1-g(y_i))\,\,\, \forall i\in I$. Observe that the reduced prices change over time. In particular, after every $d$ units of time, \alg\ samples 
 new seeds for the resources. 
Re-sampling over periods of length $d$ ensures that resources have a new seed every time they return back to the system after a match. Given the reduced prices, \alg\ matches each arrival to an available neighbor with the highest reduced price at the moment of arrival. The name Periodic Reranking comes from the following observation. When rewards $r_i=r_j\,\, \forall i,j\in I$, due to the monotonicity of $g$, the algorithm is equivalent to reranking resources after every $d$ units of time and matching arrivals to the best ranked available neighbor.


When resources are non-reusable, say $d=a(T)$, the \alg\ algorithm reduces to the Perturbed Greedy (PG) algorithm. 
For the PG algorithm, \cite{goel} showed that choosing $g(x)=e^{x-1}$ leads to the best possible guarantee of $(1-1/e)$ for OM with arbitrary rewards. In \alg, we consider the family of functions $g(x)=e^{\beta(x-1)}$ parameterized by $\beta>0$. Our analysis dictates the choice of $\beta$. In particular, $\beta =0.89$ optimizes the guarantee that can be achieved with our analysis.

In Section \ref{sec:gen}, we present a natural generalization of \alg\ for stochastic usage. Showing a performance guarantee for the generalized algorithm remains open.  
\subsection{Intuition Behind Reranking and Periodicity} 
 To gain insight into the usefulness of reranking and periodicity, consider the classic Ranking algorithm on the example below. 
\begin{eg}\label{ranker}
	\emph{	Consider a setting with two reusable resources $\{1,2\}$, identical rewards, usage duration $d>0$ and four arrivals. The first and third arrivals have edges to both resources. The second arrival only has an edge to resource 2. The fourth arrival only has an edge to resource 1. The first two arrivals occur in close proximity to each other (less than $d$ time apart). The second and third arrivals are well separated in time (more than $d$ units apart). Finally, the last two arrivals also occur close to each other (similar to the first two). 
Observe that the matching decisions at arrivals one and two have no impact on the availability of resources at arrival three. The Ranking algorithm will randomly rank the two resources. Since ranks are not changed, arrivals one and three are always matched to the same resource. Therefore, arrival two is matched if and only if arrival four is not matched. In contrast, the optimal match is obtained by ranking resource 1 over resource 2 for the first two arrivals and then \emph{reversing the ranks} for the remaining two arrivals. 
} \end{eg}  

The example does not give an upper bound on the overall performance of Ranking but illustrates a key difficulty in analyzing the performance of Ranking for reusable resources. In general, when the ranking is fixed, 
 ``right" matching decisions on early arrivals (matching the first arrival to resource 1), may imply 
 ``wrong" decisions on later arrivals (matching the third arrival to resource 1).  Reranking provides a natural way to mitigate this analytical issue. In fact, by reranking every $d$ time units, \alg\ also operationalizes the insight that a decision to match a resource at time $\tau$ does not affect the resource availability after time $\tau+d$. 
 This untangles the 
 the dependence between matching decisions for arrivals that are well separated across time and makes \alg\ tractable to analyze. 
One could also consider an extreme version of reranking, as described below. 
\smallskip

\noindent \emph{Frequent reranking}: Consider the algorithm that reranks resources at every arrival. When vertex weights are identical, say $r_i=1\,\, \forall i\in I$, this algorithm is equivalent to the following randomized algorithm: Match every arrival (that can be matched) by sampling a resource uniformly randomly. This algorithm, called Random, is known to have worst case performance same as greedy even for non-reusable resources \citep{kvv}. 
\smallskip

The key insight is that in the time span of one usage duration, each resource can be matched at most once, presenting a scenario similar to non-reusable resources. Within a period, \alg\ maintains the same rank and avoids the pitfall of frequent reranking. Indeed, \alg\ reduces to Ranking when $d\geq a(T)$. 
Finally, we discuss an alluring alternative to \alg\ that 
 reranks a resource every time it is \emph{reused}. %
\smallskip


\noindent \emph{Reranking on Return (RoR):} 
Rerank a resource every time it returns back to the system after a match. We call this the RoR algorithm. 
\smallskip

Notice that \alg\ generates a new rank more frequently than RoR. In RoR, if a resource is not highly ranked then it may not be matched and its rank is not reset. Consequently, RoR does not fully succeed in untangling dependence between matching decisions at arrivals that are well separated. Similar to Ranking, analyzing the performance of RoR remains a challenging open problem. 
\section{Analysis of Periodic Reranking}

	Our analysis relies on the primal-dual framework of \cite{devanur}, which is a versatile and general technique for proving guarantees for online matching and related problems. To describe the framework, consider the following primal problem (adapted from \cite{dickerson}), that upper bounds the optimal offline solution for OMR. 
	\begin{eqnarray}
		\textbf{Primal:}\qquad &\min\quad  &\sum_{(i,t)\in E}r_i\, x_{it} \nonumber\\
		&s.t.\ & \sum_{\tau\leq t\,\mid\, (i,\tau)\in E}\onee(a(t)-a(\tau)\leq d)\, x_{i\tau} \leq\, 1\quad \forall i\in I, t\in T\nonumber \\
		&& \sum_{i\in I\, \mid\, (i,t)\in E}x_{it} \leq 1\quad \forall t\in T\nonumber \\
		&& x_{it}\geq 0 \quad \forall (i,t)\in E\nonumber	
	\end{eqnarray} 
\noindent \textbf{Dual certificate:} 
For $t\in T$, let $t(d)$ denote the last arrival in the time interval $(a(t),a(t)+d]$. Let \alg\ denote both the algorithm and its expected total reward. Now, suppose there exist non-negative values $\lambda_t,\theta_{it}$ such that,
\begin{enumerate}[(i)]
		\item $ \alg \geq \sum_{t\in T}\lambda_t + \sum_{i\in I,\, t\in T} \theta_{it}.$ 
	\item$\lambda_t+\sum_{\tau=t}^{t(d)} 
	\theta_{i\tau}\geq\,\alpha\, r_i,\quad \forall (i,t)\in E,$
\end{enumerate}
Then, by weak LP duality, we have that PR is $\alpha$ competitive. 

Recall that the PR algorithm works in fixed periods of length $d$. Let $K=\ceil{a(T)/d}$ denote the total number of periods and let $k(t)$ denote the period that contains arrival $t\in T$. 
To ensure that $k(t)-1$ is well defined for every $t\in T$, we add a dummy period (time interval $d$) prior to the first arrival. This period does not have any arrivals and simply ensures that $k(t)\geq 2$ for every arrival. Let $y^k_i$ denote the $k$-th seed of resource $i$. Note that $y^k_i$ is the seed of $i$ in period $k\in[K]$. 
Let $Y$ denote the vector of all random seeds. Given a resource $i\in I$ and arrival $t\in T$, let $Y_{-it}$ denote the vector of all seeds except $y^{k(t)}_i$ and $y^{k(t)-1}_i$. In other words, $Y_{-it}$ captures all seeds except the seed of $i$ at $t$ and the period prior to $t$. We use $E_{y^{k(t)}_i,\, y^{k(t)-1}_i}[\cdot]$ to denote expectation with respect to the randomness in seeds $y^{k(t)}_i$ and $y^{k(t)-1}_i$. 

In order to define our dual candidate, we first define random variables $\lambda_t(Y)$, $\theta_{it}(Y)$ and subsequently set $\lambda_t=E_{Y}[\lambda_t(Y)]$ and $\theta_{it}=E_{Y}[\theta_{it}(Y)]$. Recall that $t(d)$ denotes the last arrival in the time interval $(a(t),a(t)+d]$, Inspired by \cite{devanur}, we set $\lambda_t(Y)$ and $\theta_i(Y)$ as follows. 
\smallskip

\noindent \textbf{Dual fitting:} Initialize all dual variables to 0. Conditioned on $Y$, for each match $(i,t)$ in PR set,  
 	\begin{equation}
	\lambda_t(Y)=r_i\left(1-g\left(y^{k(t)}_i\right)\right),\label{lambda}
\end{equation}
and increment $\theta_{it(d)}(Y)$ as follows,
 	\begin{equation}
\theta_{it(d)}(Y)=\theta_{it(d)}(Y)+ r_i\, g\left(y^{k(t)}_i\right). \label{theta}
\end{equation}

\begin{lemma}
	The dual candidate given by \eqref{lambda} and \eqref{theta} satisfies constraint (i) of the dual certificate.
	\end{lemma}
\begin{proof}{Proof.}
Let $\alg(Y)$ denote the matching output by PR given seed vector $Y$. From \eqref{theta}, we have 
\[\theta_{i\tau}(Y) = r_i\sum_{t\, \mid\, t(d)=\tau,\, (i,t)\in \alg(Y)}g\left(y^{k(t)}_i\right).\] 
Observe that the sets $S_{i\tau}=\{t\, \mid\, t(d)=\tau,\, (i,t)\in \alg(Y)\}\,\, \forall i\in I,\, \tau\in T$, partition the set of arrivals that are matched in $\alg(Y)$. Overloading notation, let $\alg(Y)$ also denote the total revenue of the matching $\alg(Y)$. Then,
\[\sum_{t\in T} \lambda_t(Y) + \sum_{i\in I,\, t\in T} \theta_{it}(Y) = \sum_{(i,t)\in \alg(Y)}r_i=\alg(Y). \]
\hfill\Halmos
\end{proof}

It remains to show that constraints $(ii)$ of the dual certificate hold for the desired value of $\alpha$. 
For OM, \cite{devanur} prove a stronger statement in terms of conditional expectations. We will follow a similar strategy.

\begin{lemma}
	Consider an edge $(i,t)\in E$ and seed $Y_{-it}$. Suppose that for the candidate solution given by \eqref{lambda} and \eqref{theta}, we have 
\begin{equation}\label{dual2c}
	E_{y^{k(t)}_i,\, y^{k(t)-1}_i}\left[\lambda_t(Y) + \sum_{\tau=t}^{t(d)} 
	\theta_{i\tau}(Y) \,\big|\, Y_{-it}\right]\geq \alpha\, r_i,
\end{equation}
for some value $\alpha>0$. Then, constraint (ii) of the dual certificate is satisfied for edge $(i,t)\in E$  with the same $\alpha$. 
\end{lemma}
\begin{proof} {Proof.}
	The lemma follows by taking expectation over $Y_{-it}$ on both sides of \eqref{dual2c}. 
	
	\hfill\Halmos
\end{proof}

For a given value of seed $y^{k(t)-1}_i$, resource $i$ may be matched to an arrival in period $k(t)-1$ such that it is unavailable at $t$ for all values of seed $y^{k(t)}_i$.  
Therefore, a stronger version of \eqref{dual2c} where we also fix $y_i^{k(t)-1}$ and consider a conditional expectation only w.r.t.\ random seed $y^{k(t)}_i$, does not hold for any non-trivial value of $\alpha$. This necessitates an analysis where consider an expectation w.r.t.\ both $y^{k(t)-1}_i$ and $y^{k(t)}_i$. 
In a hypothetical scenario where $i$ is available at $t$ for all values of $y^{k(t)-1}_i$, it can be shown that inequality \eqref{dual2c} holds for $\alpha=(1-1/e)$ (when $\beta=1$). Of course, in reality, $i$ may not be available at $t$ for some values of $y^{k(t)-1}_i$ 
and this scenario leads to a lower value of $\alpha$ (= 0.589) in our analysis.  

We now prove \eqref{dual2c} for every edge $(i,t)\in E$ and seed $Y_{-it}$. To this end, fix an arbitrary edge $(i,t)$ and seed $Y_{-it}$. 
To simplify notation, let $y^1=y^{k(t)-1}_i$ and $y^2=y^{k(t)}_i$. Further, let $\alg(y^1, y^2)$ denote the matching output by \alg\ given seeds $y^1,y^2$ for $i$ and with other seeds fixed according to $Y_{-it}$. Since $Y_{-it}$ is fixed, 
for simplicity,
let $\lambda_t(y^1,y^2)$ denote $\lambda_t(y^1,y^2,Y_{-it})$. Similarly, let $\theta_{i\tau}(y^1,y^2)$ denote $\theta_{i\tau}(y^1,y^2,Y_{-it})$. We also write the conditional expectation $E_{y^1,y^2}[\cdot \mid Y_{-it}]$ as $E_{y^1,y^2}[\cdot]$. 

Let $S_\tau(y^1,y^2)$ denote the set of resources available at arrival $\tau\in T$ 
in $\alg(y^1,y^2)$. 
Given $y^1\in[0,1]$, for every arrival $\tau \in T$, 
define the critical threshold $y^c_\tau(y^1)$ as the solution to,
\[r_i\left(1-g\left(y^c_\tau(y^1)\right)\right) = \max_{j\in S_\tau(y^1,1),\, (j,\tau)\in E}r_j\left(1-g\left(y^{k(\tau)}_j\right)\right). 	\]
Due to the monotonicity of function $g$, there is at most one solution to this equation. If there is no solution, we let $y^c_\tau(y^1)=0$. At a high level, the critical threshold at arrival $\tau$ captures the highest reduced price at the arrival when resource $i$ is ``removed" in period $k(t)$ (achieved by setting $y^2=1$). Similar to the analysis of OM (\cite{devanur}), this scenario serves as a foundation for establishing lower bounds on $E_{y^{1},\, y^{2}}\left[\lambda_t(y^1,y^2)\right]$ and $E_{y^{1},\, y^{2}}\left[ \sum_{\tau=t}^{t(d)} 
\theta_{i\tau}(y^1,y^2)\right]$. 

Recall that $k(t)$ is the period that contains $t$. 
Let $p(t)=[0,a(t)]\cap k(t)$ denote the sub-interval of $k(t)$ that includes all arrivals prior to (and including) $t$. We let $S_{p(t)}(y^1,1)=\cup_{\tau \in p(t)} S_\tau(y^1,1)$ i.e., $S_{p(t)}(y^1,1)$ denotes the set of all resources that are available at some point of time in interval $p(t)$. The next lemma gives useful lower bounds on 
$\lambda_t(y^1,y^2)$ and $\sum_{\tau=t}^{t(d)} \theta_{i\tau}(y^1,y^2)$, when resource $i$ is available at some point in the interval $p(t)$.

\begin{lemma}\label{devlb}
	Given $y^1\in[0,1]$ such that $i\in S_{p(t)}(y^1,1)$, we have, 
	\begin{enumerate}[(a)]
		\item $\lambda_t(y^1, y^2)\,\geq\,\lambda_t(y^1,1)\, \geq\, r_i \left(1-g(y^c_t(y^1))\right)\quad \forall y^2\in[0,1]$.
		\item $\sum_{\tau=t}^{t(d)} 
		\theta_{i\tau}(y^1,y^2) \, \geq\, \onee(y^2<y^c_t(y^1))\, r_i g(y^2)\quad \forall y^2\in[0,1].$
	\end{enumerate}
	\end{lemma}
Since every resource is matched at most once within each period, the bounds in Lemma \ref{devlb} are quite similar to their counterparts in the classic OM setting where resources are matched at most once \citep{devanur}. For a proof, see Appendix \ref{appx:missing}. 
Next, we show a useful lower bound on 
$\sum_{\tau=t}^{t(d)} \theta_{i\tau}(y^1,y^2)$ when $i\not\in S_{p(t)}(y^1,1)$. 

\begin{lemma}\label{unavail}
	Given $y^1\in[0,1]$ such that $i\not\in S_{p(t)}(y^1,1)$, we have, 
 $\sum_{\tau=t}^{t(d)} \theta_{i\tau}(y^1,y^2)\geq  r_i g(y^1)\,\, \forall y^2\in[0,1]$.
\end{lemma}
\begin{proof}{Proof.}
	Given $i\not\in S_{p(t)}(y^1,1)$, we have that $i$ 
	is matched in period $k(t)-1$ to an arrival $t'$ such that $a(t)-a(t')<d$. Let $t'(d)$ denote the last arrival in the interval $(a(t'),a(t')+d]$. Since $a(t)<a(t')+d$, we have $t'(d)\geq t$. From \eqref{theta}, given match $(i,t')$, we increase the value of $\theta_{it'(d)}(y^1, y^2)$ by $r_i\, g(y^1)$. Thus, 
	\[\sum_{\tau=t}^{t(d)} \theta_{i\tau}(y^1,y^2)\, \geq\, \theta_{it'(d)}(y^1, y^2) \, \geq\, r_i\, g(y^1)\quad \forall y^2\in[0,1].\]
	\hfill\Halmos
\end{proof}

Notice that Lemma \ref{devlb} and Lemma \ref{unavail} apply to mutually exclusive and exhaustive scenarios. Lemma \ref{devlb} applies to the case where $i$ returns in period $k(t)$ at some time prior to arrival of $t$ ($i\in S_{p(t)}(y^1,1)$). On the other hand, Lemma \ref{unavail} applies to the case where $i$ is in use during the initial part of period $k(t)$ until at least time $a(t)$ ($i\not\in S_{p(t)}(y^1,1)$). For every $y^1$, we are in one of the two scenarios. Combining Lemma \ref{devlb}$(b)$ with Lemma \ref{unavail} gives a lower bound on $\sum_{\tau=t}^{t(d)} \theta_{i\tau}(y^1,y^2)$ for all $(y^1,y^2)\in[0,1]^2$. In the proof of Lemma \ref{combine}, we turn this into a desired lower bound on the expectation $E_{y^1,y^2}\left[\sum_{\tau=t}^{t(d)} \theta_{i\tau}(y^1,y^2)\right]$. 

In the scenario where $i\in S_{p(t)}(y^1,1)$, Lemma \ref{devlb}$(a)$ lower bounds $\lambda_t(y^1,y^2)$ as a function of the critical threshold $y^c_t(y^1)$. It remains to lower bound $\lambda_t(y^1, y^2)$ when $i\not \in S_{p(t)}(y^1,1)$ and find convenient bounds on $y^t_c(y^1)$. To this end, Lemma \ref{y1} first gives a sharp characterization of the set of values of $y^1$ that lead to each scenario. Lemma \ref{connect} builds on this characterization to upper bound $y^t_c(y^1)$ in two crucial scenarios. Finally, Lemma \ref{combine} fills in the gaps and puts the various pieces together. The next lemma gives a structural result that will be used to prove Lemma \ref{y1}. 
\begin{lemma}\label{early}
Consider a value $z\in[0,1]$ such that for $y^1=z$, $i$ is matched to some arrival, say $\tau(z)$, in period $k(t)-1$. Then, for every $y^1\leq z$, $i$ is matched in period $k(t)-1$ to arrival $\tau(z)$ or an arrival that precedes it.
	\end{lemma}
\begin{proof}{Proof.}
Recall that except $y^1$ and $y^2$, all seeds are fixed. The value of $y^2$ does not affect the output of \alg\ in periods prior to $k(t)$. Similarly, the value of $y^1$ does not affect the matching prior to period $k(t)-1$. 
Since every resource can be matched at most once during a single period, when $y^1=z$, $\tau(z)$ is the unique arrival matched to $i$ during period $k(t)-1$. 

Now, let $r_i(y^1)=r_i(1-g(y^1))$ and consider the change in the matching during period $k(t)-1$ as we vary $y^1$ in the interval $(0,z)$. 
Suppose there exists a value $y^1 = z'$, with $z'<z$, such that $i$ is not matched prior to $\tau(z)$ in period $k(t)-1$ (if no such value exists, we are done). Then, for $y^1=z'$, $i$ is available at $\tau(z)$ and the matching prior to $\tau(z)$ is identical to the matching when $y^1=z$. Hence, the set of resources available at $\tau(z)$ is identical for both values of $y^1$. Since $r_i(z')\geq r_i(z)$ (by monotonicity of function $g$ for $\beta>0$), $i$ must be matched to $\tau(z)$ when $y^1=z'$. This completes the proof.
\hfill\Halmos
\end{proof}
\begin{lemma}\label{y1}
There exists values $z_1, z_2 \in [0,1]$, such that $z_1\leq z_2$ and, 
\begin{enumerate}[(a)]
	\item $i\in S_{p(t)}(y^1,1) 
	\quad \forall y^1 \in (z_2,1]$ and $i$ is not matched to any arrival in period $k(t)-1$. 
	\item $i\not\in S_{p(t)}(y^1,1)
	 \quad \forall y^1\in(z_1,z_2)$ and $i$ is matched to some arrival in period $k(t)-1$.  
	 	\item $i\in S_{p(t)}(y^1,1) 
	 \quad \forall y^1 \in [0,z_1)$ and $i$ is matched to some arrival in period $k(t)-1$. 
\end{enumerate}
	\end{lemma}
\begin{proof}{Proof.}
	Recall that $S_{p(t)}(y^1,1)$ denotes the set of all resources that are available at some point of time in interval $p(t)$. Observe that the value of $y^2$ does not influence the scenario i.e., whether $i$ is in (or not in) $S_{p(t)}(y^1,1)$. 
	
Let $z_2\in[0,1]$ be the highest value 
such that 
	for $y^1=z_2$, $i$ is matched in period $k(t)-1$. Set $z_2=0$ is no such value exists. 
	From Lemma \ref{early}, we have that for every $y^1<z_2$, $i$ will continue to be matched in period $k(t)-1$ and in fact, to (possibly) earlier arrivals. Thus, there exists a unique threshold $z_2\in(0,1)$ such that $i$ is matched in period $k(t)-1$ for every $y^1\leq z_2$, and unmatched in period $k(t)-1$ for every $y^1>z_2$. If $i$ is unmatched in period $k(t)-1$, then $i\in S_{p(t)}(y^1,1)$. This gives us part $(a)$ of the lemma. 
	
	Next, let $z_1\in[0,z_2]$ be the highest value such that $i\in S_{p(t)}(y^1,1)$ for for $y^1=z_1$. In other words, $i$ returns from its match in $k(t)-1$ in time to be available at some arrival in $p(t)$. Set $z_1=0$ if no such value exists. From Lemma \ref{early}, for every $y^1< z_1$, $i$ is matched (possibly) even earlier in period $k(t)-1$. Therefore, $i\in S_{p(t)}(y^1,1)$ for every $y^1< z_1$. This corresponds to part $(c)$ of the lemma. 
	
	Finally, by definitions of thresholds $z^1$ and $z^2$, when $y^1\in(z_1,z_2)$, we have that $y^1$ is matched in period $k(t)-1$ but $i\not\in S_{p(t)}(y^1,1)$. This corresponds to part $(b)$. 
\hfill\Halmos\end{proof}
\begin{lemma}\label{connect} The following statements are true.
	\begin{enumerate}[(a)]
		\item For every $y^1\in(z_2,1]$, 
		we have $y^c_t(y^1)= y^c_t(1)$.
		\item For every $y^1\in(z_1,z_2)$, 
		we have $y^c_t(y^1)\leq y^c_t(1)$.
	\end{enumerate}
\end{lemma}
\begin{proof}{Proof.}
From Lemma \ref{y1}$(a)$, we have that for every $y^1\in(z_2,1]$, resource $i$ is unmatched in period $k(t)-1$. Therefore, with $y^2$ fixed at 1, the matching output by \alg\ is identical for every value of $y^1>z_2$. This proves part $(a)$.

Let $r_\tau(y^1,y^2)$ denote the reduced price of the resource matched to arrival $\tau\in T$ in the matching $\alg(y^1,y^2)$. Set $r_\tau(y^1,y^2)=0$ if $\tau$ is unmatched. To prove part $(b)$, fix an arbitrary value $y^1=z\in(z_1,z_2)$ and consider the matching $\alg(z,1)$. From Lemma \ref{y1}$(b)$, we have that $i$ is matched in period $k(t)-1$ but does not return prior to arrival $t$. Let $\tau(z)$ denote the arrival matched to $i$ in period $k(t)-1$ 
and let $T'=\{\tau \in T \mid \tau'\,\leq\, \tau\,\leq\, t\}$. Observe that every arrival $\tau\in T'$, is in the interval $[a(t)-d,a(t)]$. Thus, every resource is matched to at most one arrival in $T'$. 
Now, given that $g(x)$ is strictly increasing in $x$ (for $\beta>0$), to prove $y^c_t(z)\leq y^c_t(1)$, it suffices to show that 
 $r_\tau(y^1,1)\geq r_\tau(1,1)\,\,\, \forall\, y^1\in(z_1,z_2),\,\, \tau\in T'$. Note that $1-g(1)=0$ for every $\beta$. Therefore, when $y^1=1$, the reduced price of arrival matched to $i$ is 0, same as if the arrival were unmatched. Combining this observation with the fact that \alg\ matches each arrival greedily based on reduced prices,  $r_\tau(y^1,1)\geq r_\tau(1,1)$ follows from, 
 \begin{equation}\label{nest}
 	S_\tau(1,1)\backslash\{i\}\subseteq S_\tau(y^1,1)\quad \forall y^1\in(z_1,z_2),\, \tau\in T'.
 \end{equation}
We prove \eqref{nest} via induction over the set $T'$. The first arrival in $T'$ is $\tau(z)$. From Lemma \ref{early}, prior to $\tau(z)$, resource $i$ is not matched to any arrival in period $k(t)-1$ in the matching $\alg(1,1)$. Thus, $\alg(1,1)$ and $\alg(z,1)$ are identical prior to $\tau(z)$ and $S_{\tau(z)}(1,1)= S_{\tau(z)}(z,1)$. Now, suppose that \eqref{nest} holds for all arrivals $\tau< \tau\in T'$. We show that \eqref{nest} holds for arrival $\tau'$ as well.

For the sake of contradiction, suppose there exists a resource $j\in S_{\tau'}(1,1)\backslash (S_{\tau'}(z,1)\cup \{i\})$. 
Recall that $S_{\tau}(1,1)\backslash\{i\}\subseteq S_\tau(z,1)$ for all $\tau<\tau'$. Thus, $j$ is matched to arrival $\tau'-1$ in $\alg(z,1)$, where $\tau'-1>\tau(z)$. Since every resource is matched to at most one arrival in $T'$, we have, 
$S_{\tau'}(1,1)\subseteq S_{\tau'-1}(1,1)$. Thus, $j\in S_{\tau'-1}(1,1)\backslash\{i\}\subseteq S_{\tau'-1}(z,1)$ i.e., in $\alg(1,1)$, resource $j$ is available but not matched to $\tau'-1$. This contradicts the fact that \alg\ matches greedily based on reduced prices. 

\hfill\Halmos

\end{proof}

The next lemma combines all possible scenarios given in Lemma \ref{y1} to lower bound \eqref{dual2c}. Then, the proof of Theorem \ref{ref} follows via standard algebraic arguments.

\begin{lemma}\label{combine}
	Let $g(x)=e^{\beta(x-1)}$ for some $\beta\in(0,1]$. Let $G(x)$ be the antiderivative of $g(x)$. 
	There exists values $z_1, z_2\in[0,1]$ with $z_1\leq z_2$, such that
\begin{eqnarray}
	&&E_{y^1, y^2}\left[\lambda_t(y^1,y^2) + \sum_{\tau=t}^{t(d)} 
	\theta_{i\tau}(y^1,y^2)\right]\nonumber\\
	&&\quad \geq  r_i\, \Bigg[G(z_2)-G(z_1)+ (1-g(y^c_t(1)))(1-z_1)+(1-z_2)\left(G(y^c_t(1))-G(0)\right)
	+z_1(1-g(0))
	\Bigg].\nonumber
\end{eqnarray}
\end{lemma}
\begin{proof}{Proof.} Observe that for any random variable $X$ derived from $y^1, y^2$, we have,
	\[E_{y^1,y^2}[X] = (1-z_2)\, E_{y^1,y^2}\left[X\mid y^1>z_2\right]+(z_2-z_1)\, E_{y^1,y^2}\left[X\mid y^1\in(z_1,z_2)\right]+z_1\, E_{y^1,y^2}\left[X\mid y^1<z_1\right].\]
	We prove the main claim by establishing lower bounds on each of the three terms on the RHS. 
\smallskip

	\noindent \textbf{Case I:} $\mb{y^1>z_2}$.
From Lemma \ref{y1}$(a)$, when $y^1>z_2$, $i$ is not matched in period $k(t)-1$. Thus, $\alg(y^1,y^2)=\alg(1,y^2)\,\,\,\forall y^1>z_2$ i.e., in this case, the matching does not change with $y^1$. From Lemma \ref{devlb}$(a)$ and Lemma \ref{connect}$(a)$, we have
\[\lambda_t(y^1,y_2)\geq r_i\,\left(1-g(y^c_t(1))\right)\quad \forall y^1\in(z_2,1],\, y^2\in[0,1].\]
Taking expectation over randomness in $y^1, y^2$, we have 
\[E_{y^1}\left[E_{y^2}[\lambda_t(y^1,y^2) \mid y^1> z_2]\right]\geq E_{y^1}\left[r_i\, \left(1-g(y^c_t(1))\right)\right] = r_i\,\left(1-g(y^c_t(1))\right).\]
 Finally, from Lemma \ref{devlb}$(b)$ and Lemma \ref{connect}$(a)$, we have,
\begin{eqnarray*}
	E_{y^1}\left[E_{y^2}\left[\sum_{\tau=t}^{t(d)}\theta_{i\tau}(y^1,y^2)\mid y^1>z_2\right]\right] &\geq &E_{y^1}\left[E_{y^2}[\onee(y^2<y^c_t(1))\, r_i g(y^2)]\mid y^1>z_2\right],\\
	&= &E_{y^1}\left[ r_i\, \int_0^{y^c_t(1)} g(x) dx\mid y^1>z_2\right],\\
	&= &r_i\,(G(y^c_t(1))-G(0)). 
	\end{eqnarray*}

\noindent \textbf{Case II:} $\mb{z_1<y^1<z_2}$. In this case, $i$ is not available in period $k(t)$ prior to arrival $t$ and the value of $y^2$ does not affect the matching until after arrival $t$. From Lemma \ref{unavail}, we have,
\[E_{y^1,y^2}\left[\sum_{\tau=t}^{t(d)} \theta_{i\tau}(y^1,y^2)\mid y^1\in(z_1,z_2)\right]\geq  r_i \int_{z_1}^{z_2} g(x)dx\,=\, r_i\, (G(z_2)-G(z_1)).\]
From Lemma \ref{connect}$(b)$, we have
\[E_{y^1,y^2}\left[\lambda_t(y^1,y^2)\mid y^1\in(z_1,z_2)\right]\,=\,E_{y^1}\left[\lambda_t(y^1,1)\mid y^1\in(z_1,z_2)\right]\,\geq\, r_i\, \left(1-g(y^c_t(1))\right). \]

\noindent \textbf{Case III:} $\mb{y^1<z_1}$. 
In this case, $i\in S_{p(t)}(y^1,1)$. From Lemma \ref{devlb}$(a)$, we have
\begin{equation}\label{c3l}
	E_{y^2}[\lambda_t(y^1,y^2)\mid y^1<z_1]\geq r_i\, \left(1-g(y^c_t(y^1))\right).
\end{equation}
From part $(b)$ of Lemma \ref{devlb}, 
\begin{equation}\label{c3t}
		E_{y^2}\left[\sum_{\tau=t}^{t(d)}\theta_{i\tau}(y^1,y^2)\mid y^1<z_1\right] \geq E_{y^2}\left[\onee(y^2<y^c_t(y^1))\, r_i g(y^2)\mid y^1<z_1\right]\,=\, r_i\,\left(G(y^c_t(y^1))-G(0)\right).
	\end{equation}
Combining \eqref{c3l} and \eqref{c3t}, we have,
\begin{eqnarray*}
	E_{y^2}\left[\lambda_t(y^1,y^2)+\sum_{\tau=t}^{t(d)}\theta_{i\tau}(y^1,y^2)\mid y^1<z_1\right] &\geq &r_i\left[1-g(y^c_t(y^1))+G(y^c_t(y^1))-G(0)\right],\\
	&= &r_i\left[1-g(y^c_t(y^1))+\frac{1}{\beta}\left(g(y^c_t(y^1))-g(0)\right)\right],\\
	&\geq &r_i \min_{x\in[0,1]}\left[1-g(x)+\frac{1}{\beta}\left(g(x)-g(0)\right)\right],\\
	&= &r_i\, (1-g(0)).
	\end{eqnarray*}
The first equality uses the fact that $G(x)=\frac{1}{\beta}g(x) + c$, where $c$ is some constant. The second equality follows from the fact that $h(x)=1-g(x)+\frac{1}{\beta}\left(g(x)-g(0)\right)$, is a non-decreasing function of $x$ for $\beta\in(0,1]$. Thus,
\[E_{y^1}\left[	E_{y^2}\left[\lambda_t(y^1,y^2)+\sum_{\tau=t}^{t(d)}\theta_{i\tau}(y^1,y^2)\mid y^1<z_1\right]\right]\geq  r_i\, (1-g(0)). \]
\smallskip
\hfill\Halmos\end{proof}

\begin{proof}{Proof of Theorem \ref{ref}.} 	Let 
	\[f(z_1,z_2,x)=G(z_2)-G(z_1)+ (1-g(x))(1-z_1)+(1-z_2)\left(G(x)-G(0)\right)	+z_1(1-g(0)).\]
	 We show that, $\min_{0\leq z_1\leq z_2\leq 1,\, x\in[0,1]} f(z_1,z_2,x)> 0.589$ for $\beta=0.89$. Then, using Lemma \ref{combine} completes the proof. First, using the fact that $G(x)=\frac{1}{\beta}g(x)+c$, where $c$ is some constant, we have,
	\begin{eqnarray}
		f(z_1,z_2,x)
		&= &\frac{1}{\beta}\left(g(z_2)-g(z_1)\right)+ (1-z_1)\left(1-g(x)\right)+\frac{1-z_2}{\beta}\left(g(x)-g(0)\right)+z_1\left(1-g(0)\right)\nonumber\\
		&= &\frac{1}{\beta}\left(g(z_2)-g(z_1)\right)+1-\frac{g(0)}{\beta}\left(1-z_2+\beta z_1\right)+\frac{g(x)}{\beta}\left( 1-z_2+\beta z_1-\beta\right)\label{exp1}
	\end{eqnarray}
To find the minimum of this function, consider the following cases.
\smallskip

\noindent \textbf{Case I: $\mb{1-z_2\geq \beta(1-z_1)}$.} In this case, \eqref{exp1} is minimized at $x=0$. Thus, 
	\begin{eqnarray*}
	f(z_1,z_2,x)&\geq &\frac{1}{\beta}\left(g(z_2)-g(z_1)\right)+1-g(0)\\
	&\geq & 1-g(0)\, =\, 1-e^{-\beta}, 
\end{eqnarray*}
where we used the fact that $g(z_2)\geq g(z_1)$ for $\beta\geq 0$ and $z_2\geq z_1$. 
\smallskip

\noindent \textbf{Case II: $\mb{1-z_2< \beta(1-z_1)}$.} In this case \eqref{exp1} is minimized at $x=1$. Thus,
	\begin{eqnarray*}
	f(z_1,z_2,x)&\geq &\frac{1}{\beta}\left(g(z_2)-g(z_1)\right) +\frac{1-g(0)}{\beta}\left( 1-z_2+\beta z_1\right)\label{exp2}
\end{eqnarray*}
Observe that the function $-e^{\beta z_1}+c\, z_1$, where $c$ is some constant, is concave in $z_1$. Thus, \eqref{exp2} is minimized at $z_1=0$ or $z_1=\min\left\{z_2,1-\frac{1-z_2}{\beta}\right\}$. In fact, $z_2\geq 1-\frac{1-z_2}{\beta}$ for every $\beta\leq 1$. Therefore,
	\begin{eqnarray*}
	f(z_1,z_2,x)&\geq &\min\left\{\frac{1}{\beta}\left(e^{\beta(z_2-1)}-e^{-\beta}\right) +\frac{1-e^{-\beta}}{\beta}\left( 1-z_2\right),\,\, \frac{1}{\beta}\left(e^{\beta(z_2-1)}-e^{z_2-1}\right) +1-e^{-\beta}\right\}\\
	&\geq &\min\left\{\frac{1}{\beta}\left(e^{\beta(z_2-1)}-e^{-\beta}\right) +\frac{1-e^{-\beta}}{\beta}\left( 1-z_2\right),\,\, 1-e^{-\beta}\right\},
\end{eqnarray*}
where we used the fact that $e^{\beta(z_2-1)}-e^{z_2-1}\geq 0$ for $\beta\leq 1$ and $z_2\in[0,1]$. 

Combining both cases, we have that
\[	f(z_1,z_2,x)\geq \min\left\{\min_{z_2\in[0,1]}\frac{1}{\beta}\left(e^{\beta(z_2-1)}-e^{-\beta}\right) +\frac{1-e^{-\beta}}{\beta}\left( 1-z_2\right),\,\, 1-e^{-\beta}\right\}.\]
The first term inside the minimum is the solution to a convex minimization problem. It is easy to (numerically) verify that when $\beta=0.89$, both terms are greater than $0.5893$, giving us the desired guarantee. We numerically tried different values of $\beta$ to arrive at the conclusion that $0.89$ is the best choice for $\beta$. 
\begin{figure}[h]\label{plot}
	\includegraphics[scale=0.4]{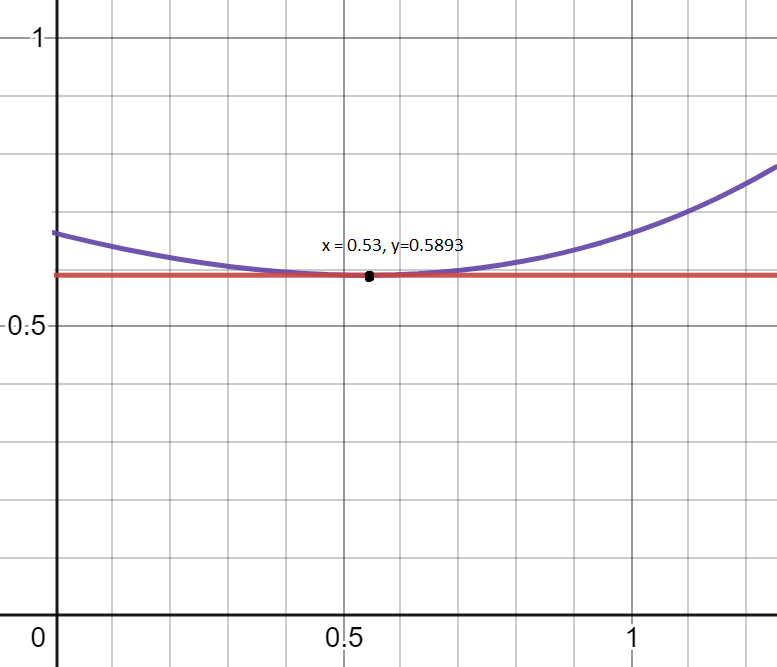}
	\centering
	\caption{With $\beta=0.89$, the straight line (red) is $y=1-e^{-\beta}> 0.5893$ and the curve (purple) is $y=\frac{1}{\beta}\left(e^{\beta(x-1)}-e^{-\beta}\right)+\frac{1-e^{-\beta}}{\beta}(1-x)\geq 0.5893$. Other values of $\beta$ lead to a lower minimum point. }
\end{figure}
When $\beta=1$, we have a minimum value of $0.554$.
\hfill\Halmos\end{proof}
\section{Reranking for Stochastic Usage}\label{sec:gen}
So far, we focused on a setting with deterministic and identical usage duration for all resources. In general, resources may be used for different durations and the time of use may vary with every match. Indeed, previous work models these scenarios by considering (known) usage \emph{distributions} $F_i\,\,\, \forall i\in I$. In the general model, when resource $i$ is matched, it is used for an independently sampled random duration $d_i\sim F_i$. To generalize \alg\ in this setting, we use insights from previous work~\citep{full} on the stochastic usage model. 
\smallskip

\noindent \textbf{$\mb{(F_i,\sigma)}$ random process~\citep{full}:}  Consider an ordered set of points $\mb{\sigma}=\{\sigma_1,\cdots,\sigma_{T}\}$ on the positive real line such that, $0<\sigma_1< \sigma_2< \cdots< \sigma_T$. These points are also referred to as arrivals. We are given a single unit of a resource that is in one of two states at every point in time: \emph{free/available} or \emph{in-use/unavailable}. 

The unit starts at time $0$ in the available state. The state of the unit evolves with time as follows. If the unit is available at arrival $\sigma_t\in \mb{\sigma}$, it switches to in-use for an independently drawn random duration $d\sim F_i$. The unit stays in-use during $(\sigma_t,\sigma_t+d)$ and switches back to being available at time $\sigma_t+d$.
\smallskip

For $i\in I$, let $A_i=\{a(t)\mid (i,t)\in E,\, t\in T \}$ denote the ordered set of arrival times for every arrival with an edge to $i$. Consider the $(F_i, A_i)$ random process and let $\eta_i(\tau)$ denote the probability that the resource is in available state at time $a(\tau)$ in process $(F_i, A_i)$. Using this, we propose the following generalization of \alg.


\begin{algorithm}[H]
	\SetAlgoNoLine
	\textbf{Inputs:} Set of resources $I$, usage distributions $\{F_i\}_{i\in I}$, parameter $\beta$\; 
	Let $g(t)=e^{\beta(t-1)}$ and $S=I$\;
	\For{\text{every new arrival } $t$}{
			\smallskip
		\For{$(i,t)\in E$}{
			 Generate new i.i.d. rank $z_i(t)\sim U[0,1]$\;
			 $y_i= \sum_{\tau\leq t}  \left(1-F_i\left(a(t)-a(\tau)\right)\right)\, \eta_i(\tau)\,z_i(\tau)$\;
		}
		\smallskip
		Update set $S$ by adding resources that returned since arrival $t-1$\;
		Match $t$ to $i^*=\underset{ i\in S,\, (i,t)\in E}{\arg\max}\quad r_i (1-g(y_i))$\;
		$S=S\backslash \{i\}$\;
	}	
	\caption{Fluid Reranking}
	\label{freerank}
\end{algorithm}
When usage durations are deterministic and identical, it can be shown that Algorithm \ref{freerank} is equivalent to \alg. We conjecture that this generalization beats greedy for the general model of OMR with stochastic reusability. 
\section{Conclusion} 
We considered a fundamental generalization of classic bipartite online matching, where resources are reusable and used for an (identical) deterministic duration on every match. We motivated and introduced a new algorithm, Periodic Reranking, that strikes a careful balance between greedy and Ranking by reranking resources on a periodic schedule. We established a guarantee of 0.589 for this algorithm using the well known primal-dual method of \cite{devanur}. To apply this framework in an effective way, we showed various novel structural properties of our algorithm. 
Finally, we proposed a generalization of our algorithm to settings with stochastic reusability. Proving a performance guarantee better than 0.5 in the stochastic model remains an open problem for unit inventory settings.
\ACKNOWLEDGMENT{The author thanks Vineet Goyal and Garud Iyengar for many insightful discussions on this topic. The setting of identical deterministic usage durations was  posed as an open problem by Rad Niazadeh during a talk at the INFORMS Revenue Management and Pricing Conference (2021).}
	{\small
	\bibliographystyle{informs2014.bst}
	\bibliography{bib.bib}
}
\begin{APPENDICES}
\section{Missing Proofs}\label{appx:missing}
\begin{repeatlemma}[Lemma \ref{devlb}.]
	Given $y^1\in[0,1]$ such that $i\in S_{p(t)}(y^1,1)$, we have, 
	\begin{enumerate}[(a)]
		\item $\lambda_t(y^1, y^2)\,\geq\,\lambda_t(y^1,1)\, \geq\, r_i \left(1-g(y^c_t(y^1))\right)\quad \forall y^2\in[0,1]$.
		\item $\sum_{\tau=t}^{t(d)} 
		\theta_{i\tau}(y^1,y^2) \, \geq\, \onee(y^2<y^c_t(y^1))\, r_i g(y^2)\quad \forall y^2\in[0,1].$
	\end{enumerate}
\end{repeatlemma}
\begin{proof}{Proof.} 
	
	\emph{Part(a)}: By definition of $y^c_t(y^1)$, we have $\lambda_t(y^1,1)\, \geq\, r_i \left(1-g(y^c_t(y^1))\right)$. It remains to show that, $\lambda_t(y^1, y^2)\,\geq\,\lambda_t(y^1,1)$. Let $r_\tau(y^1,y^2)$ denote the reduced price of the resource matched to arrival $\tau\in T$ in the matching $\alg(y^1,y^2)$. Set $r_\tau(y^1,y^2)=0$ if $\tau$ is unmatched. 
	
	Since $g(x)$ is strictly increasing in $x$ (for $\beta>0$), 
	it suffices to show that 
	$r_\tau(y^1,y^2)\geq r_\tau(y^1,1)\,\,\, \forall\, y^2\in[0,1],\,\, \tau\in k(t)$. 	Note that $1-g(1)=0$ for every $\beta$. Therefore, when $y^2=1$ the reduced price of arrival matched to $i$ is 0, same as, if the arrival were unmatched. Combining this observation with the fact that \alg\ matches each arrival greedily based on reduced prices, it suffices to show that,
	\begin{equation}\label{neste}
		S_\tau(y^1,1)\backslash\{i\}\subseteq S_\tau(y^1,y^2)\quad \forall y^2\in[0,1],\, \tau\in k(t).
	\end{equation}

	We prove \eqref{neste} via induction over arrivals in period $k(t)$. Let $t(y^2)$ denote the first arrival in period $k(t)$ where $i$ is available. Matchings $\alg(y^1,1)$ and $\alg(y^1,y^2)$ are identical prior to $t(y^2)$. Thus, $S_{t(y^2)}(y^1,1)= S_{t(y^2)}(y^1,y^2)$. Now, suppose that \eqref{neste} holds for all arrivals $\tau< \tau'\in k(t)$. We show that \eqref{neste} holds for arrival $\tau'$ as well.
	
	For the sake of contradiction, suppose there exists a resource $j\in S_{\tau'}(y^1,1)\backslash (S_{\tau'}(y^1,y^2)\cup \{i\})$.	Recall that $S_{\tau}(y^1,1)\backslash\{i\}\subseteq S_\tau(y^1,y^2)$ for all $\tau<\tau'$. 
	This occurs only if $j$ is matched to arrival $\tau'-1$ in $\alg(y^1,y^2)$, where $\tau'-1\geq t(y^2)$. 
	Since every resource is matched at most once in a period, we have, $S_{\tau'}(y^1,1)\subseteq S_{\tau'-1}(y^1,1)$. Thus, $j\in S_{\tau'-1}(y^1,1)\backslash\{i\}\subseteq S_{\tau'-1}(y^1,y^2)$ i.e., in $\alg(y^1,1)$, resource $j$ is available but not matched to $\tau'-1$. This contradicts the fact that \alg\ matches greedily based on reduced prices.
	\smallskip
	
	\noindent \emph{Part(b)}: Let $T'$ denote the set of arrivals in the interval $k(t)\cap [a(t)-d,a(t)]$. We are given that resource is available at some point in this interval. Further, every resource is matched to at most one arrival in $T'$. Consider an arbitrary value $y^2=z\in[0,1]$ and the matching $\alg(y^1,z)$.

First, if $i$ is matched to an arrival in $T'$, then by definition of dual variables \eqref{theta} and $t(d)$, we have that, $\sum_{\tau=t}^{t(d)} \theta_{i\tau}(y^1,z) \geq r_i g(z)$. On the other hand, if $i$ is not matched to any arrival in $T'$, then, $\alg(y^1,z)$ is identical to $\alg(y^1,1)$ until after arrival $t$. Therefore, in $\alg(y^1,z)$, $t$ is matched to the same resource as in $\alg(y^1,1)$, despite $i$ being available. Thus, 
$r_i(1-g(z))<r_i(1-g(y^c_t(y^1)))$ i.e., $z>y^c_t(y^1)$.

\hfill\Halmos\end{proof}
\end{APPENDICES}
	\end{document}